\documentclass[format=acmsmall]{acmart}
\usepackage{acm-ec-like}
\usepackage{booktabs} 
\usepackage[ruled]{algorithm2e} 

\SetAlFnt{\small}
\SetAlCapFnt{\small}
\SetAlCapNameFnt{\small}
\SetAlCapHSkip{0pt}
\IncMargin{-\parindent}

\setcitestyle{authoryear}

\usepackage{hyperref}
\usepackage{float}
\usepackage{amsfonts}
\usepackage{subfigure}
\usepackage{amsmath}
\usepackage{mathrsfs}
\usepackage{enumerate}
\usepackage{amsthm}
\usepackage{dsfont}
\usepackage{cleveref}
\usepackage{verbatim}
\usepackage{amssymb}

\usepackage{color}
\usepackage{xcolor}

\newcommand{\argmax}{\operatornamewithlimits{argmax}}
\newcommand\cdotfill{%
	\leavevmode\cleaders\hb@xt@.44em{\hss$\cdot$\hss}\hfill\kern\z@
}

\theoremstyle{remark}
\newtheorem*{Rem}{Remark}

\newcommand{\R}{\mathbb{R}}
\newcommand{\N}{\mathbb{N}}

\newcommand{\EE}{\mathbb{E}}

\title[Pricing under network effects]{Pricing under a multinomial logit model with  non linear network effects.}

\author{Felipe Maldonado}
\email{felipe.maldonado@in.tum.de}

\author{Gerardo Berbeglia}
\author{Pascal Van Hentenryck}

\begin{abstract}
We study the problem of pricing under a Multinomial Logit model where we incorporate network effects over the consumer's decisions. We analyse both cases, when sellers compete or collaborate. In particular, we pay special attention to the overall expected revenue and how the behaviour of the no purchase option is affected under variations of a network effect parameter. Where for example we prove that the market share for the no purchase option, is decreasing in terms of the value of the network effect, meaning that stronger communication among costumers increases the expected amount of sales.
We also analyse how the customer's utility is altered when network effects are incorporated into the market, comparing the cases where both competitive and monopolistic prices are displayed. We use tools from stochastic approximation algorithms to prove that the probability of purchasing the available products converges to a unique stationary distribution. We model that the sellers can  use this stationary distribution to establish their strategies. Finding that under those settings, a pure Nash Equilibrium represents the pricing strategies in the case of competition, and an optimal (that maximises the total revenue) fixed price characterise the case of collaboration.
\end{abstract}

\begin{document}

\maketitle

\section{Introduction}
The widespread use of internet has created many new type of markets that are reshaping the global economy, for example, people now watch movies on {\sc Netflix} instead of renting a DVD at {\sc Blockbuster}. These Internet-based markets do not necessary follow the same rules than traditional markets (which have been well studied for decades), since their structure can be fairly different, where for example products can have unlimited supply (e.g., digital goods like songs), and millions of users from all across the world can access to them instantaneously. All these new type of markets open research opportunities
in many disciplines such as Economics, Operations Research and Computer Science, where researchers could tackle problems like novel pricing schemes, subscription-based fees, recommendations systems and many more.

A very interesting feature of these markets is the effect of consumption history, reflected in a {\it social signal} (e.g., 5 stars rating, number of views, etc.), over the purchasing decisions of upcoming customers, phenomenon that in the literature is referred as {\bf network effects}. Consumers make their purchasing decisions (choose one product over the others, or do not purchase anything) not only based on the quality and prices of the available alternatives, but also based on market-specific features such as rating systems that keep track of past consumption and opinions. These network effects  become even more relevant when the prior information about the products is scarce, so the willingness to try/pay is heavily influenced by the opinion of the rest ({\it Wisdom-of-the-Crowd} e.g., \cite{wang2014quantifying}).

%

In this paper we aim to study seller's pricing strategies based on a model of consumer choice, where the purchasing decisions are affected by past consumption. In general terms we assume that the willingness to purchase is influenced by the (known) intrinsic utility of the products, their prices, and network effects as a function of consumption history. The consumers can purchase the product $i\in \{1,\dots, n\}$ that maximises their expected utility or they can choose to leave the market without making any purchase, what is called choosing the {\it no purchase option}. 

We also aim to represent how ineffective transactions affect the purchasing decisions, motivated for markets like {\tt eBay}, where after each transaction the users can give an evaluation to the seller in the categories of positive, negative and neutral. New consumers can observe how many transactions a seller has made and a reputation score that penalises the negative feedback, providing  extra information about transactions where the consumers were not satisfied (e.g, \cite{cabral2010dynamics}). In order to capture those ineffective transactions with our model, we consider that the no purchase option also presents network effects. In this way we are able to keep track of consumers that do not buy anything in the market, or equivalently they buy similar products somewhere else.

The main contributions of this paper can be summarised in the following way:

\begin{itemize}
	\item {\bf Non-linear network effects in a consumer choice model:} We propose a variation of the Multinomial Logit Model  for consumer choice where we incorporate non-linear  network effects, representing in this way, market interactions where consumers only see a score function based on past consumption.  Since the  probability of choosing the available products (or the no purchase option)  dynamically changes over time due to the network effects, we apply stochastic approximation techniques to prove that  such probability converges almost surely to an asymptotic stationary distribution, that represents the market share of each product in the long run.
	
	\item {\bf Monopolistic and competitive pricing are analysed:} For a market with $n$ sellers, we model their expected revenues based on the asymptotic market share distribution and the displayed prices. First, we study  the case where  sellers act collaboratively,  adopting a monopolistic pricing strategy to maximise the overall expected revenue. We show that the market share of the no purchase option is decreasing in terms of the network parameter $r: 0<r<1$, and that the overall expected revenue is increasing in that parameter as long as $r$ is large enough. 
	We then study the case where  sellers compete,  inducing a price competition game that has  a unique pure Nash Equilibrium (we also provide an algorithm to compute it). We finally compare experimentally and analytically both cases, incorporating the consumers' perspective into the analysis.
\end{itemize}

The rest of the paper is structured as follows: Section \ref{sec:RelLit} describes the related literature. Section \ref{sec:model} details the proposed consumer choice model. Section \ref{sec:Mono} focuses on the sellers acting collaboratively, in opposition  to Section \ref{sec:compe} where they compete. Finally, Section \ref{sec:Compare} compares the strategies defined in the previous two sections, showing how they affect/benefit the consumers.

Across the whole paper we include some numerical examples based on synthetic data, complementing the theoretical results and providing some extra insights. We also include an Appendix where more experiments are shown.


\subsection{Related Literature}
\label{sec:RelLit}

Capturing the way people make decisions has been a problem of interest across different disciplines for many decades, having on one hand classic models from Economic Theory, and on the other hand data-driven approaches from Machine Learning, two different perspectives that aim to the same: understand consumer behaviour, and eventually predict with certain accuracy future outcomes. Many features have been considered into these models (e.g., type of users, willingness to pay, etc.), trying to establish what is more relevant to the consumers, leading to better structured markets. 

Predicting the sales quantities is a key element in the field of Revenue Management, where sellers have to decide  what products to sell and their best prices that maximise their revenues (among other decisions). In order to do that, it is required to have at least an estimation of the consumers' demands for each one of the products. Such a problem has been widely studied in Economics, where classic models  assume that each user obtains certain utility for buying a particular product (given by a real number), so among all the available discrete options, consumers try to maximise their utilities. It is important to notice that this can be as general as possible, where for instance, among the available options we can consider bundles of products as a single one, or include the no-purchase option as a fictional product that captures the consumers that do not buy anything.  

An important subclass of discrete choice models are the Random Utility Models (RUM) (\cite{block1959random}). Among the most common Random Utility Models, we can find the  Multinomial Logit (MNL) model, originally introduced by \cite{luce1959individual}, widely used in fields such as Marketing, Economics, and Computer Science, where it is often used for operational and managerial decisions problems such as assortment optimisation (\cite{wang2016consumer}), pricing (\cite{besbes2016product}), scheduling (\cite{feldman2014appointment}).

The MNL model has many advantages due to the simplicity on how it is defined, leading to desirable results like being computational tractable (e.g., assortment optimisation can be computed in polynomial time \cite{talluri2004revenue,rusmevichientong2010dynamic}), however it exhibits the property known as independence of irrelevant alternatives (IIA), which states that the ratio of the probabilities of being chosen between two alternatives, is independent of the rest of alternatives. In practice, this property is often violated, particularly when there are more than two similar alternatives. To overcome this limitation, several extensions have been proposed, among them we can find the Nested Multinomial Logit (NMNL) model (\cite{williams1977formation}), where the alternatives  are grouped in nests, choosing each nest follows a MNL model, and choosing each alternative within each nest, is also chosen accordingly a MNL. Another extension is the Mixed Multinomial Logit (MMNL) model  (\cite{daly1978improved}) that considers random utilities and integrates the original MNL model over the distribution of utilities. Some of the downsides of these more general choice models is the computational complexity associated to them,  while problems like assortment (choosing the subset of products that maximise the expected revenue) under the MNL model admits a polynomial-time algorithm, in the case when consumers follow either a NMNL or MMNL model, the optimal assortment problem   is NP-hard (\cite{davis2014assortment, rusmevichientong2010assortment} respectively).


As \cite{berbeglia2018generalized} states, RUM's fail to explain several choice phenomena, such as the {\it decoy effect}, where the inclusion of a similar but inferior product  into the option set, can increase the probabilities of being chosen for some of the original products (a typical example is to include a medium size popcorn with a price close to the large size option). Hence, more complex consumer behaviour has led to the  inclusion of more general choice models that are not RUM's such as the Perception-Adjusted Luce model (PALM) (\cite{echenique2018perception}), the General Attraction Model (GAM) (\cite{gallego2014general}), the General Luce Model (\cite{echenique2015general}), and the  General Stochastic Preference \cite{berbeglia2018generalized}). PALM for example considers a perception effect, where the individuals check sequentialy their alternatives according a perception priority order, and the probability of choosing an alternative is affected by the probability of not choosing alternatives with higher priority.

A big part  of the problems studied with the use of discrete choice models are either assortment  or pricing problems (or a combination of both), where researchers weight the trade-off between having a more general model and its computational complexity. The different types of problem studied under these models has lead to many extensions, in \cite{besbes2016product} for example, the authors   present a model where the demands follow a MNL model, they analyse equilibrium outcomes when different firms compete and face a display constraint (assortment), each retailer needs to choose strategically which products to show and what prices in order to maximise their revenues. The analysis is separated in two parts, the first part is when the prices are fixed by an external agent and the firms only compete  in assortments, and the second is when they compete on both assortment and pricing, for the latter the prices are chosen according an assortment maximisation strategy.
\cite{li2011pricing} on the other hand,  study  a case where a Nested Logit model (including MNL as a special case, when there is only one nest) represents the demands of consumers, defining what the authors call the {\it market share } of the products, the authors find an optimal price, that maximises the revenue for a monopolist selling multiple products. A price and quantity competition are also studied under simpler conditions for the case of an oligopoly.

Some recent research have also incorporated the effect of past purchases (network effects) into a MNL model for consumer choice, for example in \cite{wang2016consumer} and \cite{du2016optimal}, the authors propose a model that focuses  in a monopolistic environment (studying assortment, and pricing optimisation respectively), defining a consumer utility function affected linearly by network effects. Their models has led to many related research and extensions (e.g., \cite{cui2016exact,chen2017duopoly}). 

Most of the previous research that include network effects into their models have focused on monopolistic markets, among the exceptions  we can find  \cite{li2011pricing} and \cite{chen2017duopoly}, where the latter analyses a duopoly in which the firms compete using the market share as a decision variable (instead of the price),  finding multiple Nash Equilibria depending on the strength of the network effects and the quality of the products. In contrast to them, as we will see in Section \ref{sec:compe}, when we study a competition between sellers,  our model leads to a unique Nash Equilibrium.

Also in the competition research literature, we can find a recent paper, \cite{feng2017blockbuster} where the authors provide  a game theoretical approach to a market where the strategic sellers decide to enter if their expected revenues are positive, their managerial decision is the investment in the quality of their products. After the quality game is played, sequential customers enters to the market and base their purchasing decisions on the qualities and the current sales volume.  Unlike the model presented in Section \ref{sec:model} , they do not consider a no purchase option, since they assume the prices are the same for every product and fixed beforehand.

In a different stream of literature  some researches have focussed on social networks and pricing decisions over the services provided (e.g.,  \cite{saaskilahti2015monopoly, chen2011optimal,candogan2010optimal,crapis2016monopoly} ), due to the nature of this type of network, most of the research in this area only analyses monopolistic pricing. However  some extra complexities have also been included into their models , such as incomplete information.  For example \cite{crapis2016monopoly} considers a model where the qualities of the products have a random distortion, and the preferences for each product follow a known distribution, the author study the monopolist's pricing problem where  sequential customers arrive and face the decision of buying or taking an outside option. Under some conditions based on social interactions, the products' qualities eventually can be learnt, and under this setting two pricing policies are proposed (static price, and single change price).

From a model perspective, the closest papers to our research are  \cite{du2016optimal} and  \cite{cui2016exact}, where the authors have  among their results, that for the homogeneous case (identical products), if the network effects are strong enough then the optimal price assigns the same price to all the products except for one (arbitrary) product, which gets a lower price. This result differs from the classical MNL model without network effects, where in such a case, all the products have the same price. In our model on the other hand, the presence of network effects does not affect that outcome, obtaining the same price for all products. That price depends on a network parameter $r$, and when $r\to 0$ we recover the prices from the MNL without network effects.  Another key difference is that unlike their market models (which are a static ), we present a dynamic model where customers arrive sequentially and observe a different network signal based on previous purchases, and the way the market shares are updated (stepsize) does not affect the long run behaviour. 

\section{Model}
\label{sec:model}
We consider a market with $n$  sellers,  $n\geq 2$, where each seller $i\in \{1,...,n\}$ owns one indivisible product with unlimited supply (e.g., digital goods like e-books).  For notational convenience we also call $i$  to the product of seller $i$.



Once the sellers have fixed the prices for their products, sequential consumers arrive and decide to buy one of the $n$ products or not to buy anything. We define a discrete time $k\geq 1$ as the arrival of consumer $k$ to the market. 

We assume that consumers' decisions are affected by the intrinsic utility of the products, the prices and some network effects related to the popularity of the products.
We model this as a variation of a standard MNL model with network effects (see for example \cite{du2016optimal,du2018optimal}), where we incorporate a non-linear network effect (in the consumers' utilities), reflecting a score function based on past purchases. This is done mainly for two reasons: first, we intend to use some results from prior related research (e.g., \cite{maldonado2018}), that gives us some assurances over the asymptotic behaviour of the market, and second, we aim to avoid multiplicity of price equilibria, a phenomenon that can be observed for example in   \cite{du2016optimal}. 

Formally our model is defined as follows:  the $k+1$-th consumer's utility obtained from purchasing product $i$ is given by
\begin{equation}
	\label{customerUt}
	u_i^{k}:=u_i(r,g_i, d_i^k, p_i)=g_i+r\ln(d_i^k)-\beta_ip_i+\xi_i,
\end{equation}

where $g_i$ represents the intrinsic utility of product $i$ (a measure of its quality); $r$ is a constant that represents the strength of the network effect on the consumers ($0<r<1$); $p_i$ is the price of the product $i$ and $\beta_i$ its price sensitivity; $d_i^k$ is its cumulative amount of purchases up to time $k$,  that for notational convenience we initialised  as $d_i^0=1$ for all $ i\in \{1,\dots, n\}$ (this is equivalent to consider $\ln(d_i^k+1)$, with $d_i^0=0$ ). Finally $\xi_i$ is a random variable representing consumer specific idiosyncrasies.

We also consider a {\it dummy} product, $n+1$, representing the no purchase option, which we characterise  with the parameters $ g_{n+1}=0$,  $p_{n+1}=0$, and $d_{n+1}^0=1$ which is increased by 1 every time a new consumer does not buy anything, keeping track of the ineffective exchanges between  sellers and consumers. The utility for the no purchase option is then $u_{n+1}^{k}=r\ln(d_{n+1}^k)+\xi_{n+1}$. 

In \cite{dhar1997consumer} the author shows several empirical studies  where consumers decide for a no purchase option, even when the available products have a good intrinsic utility. With our model we try to capture that type of phenomenon.  

We denote $ [n+1]:=\{1,...,n\}\cup\{n+1\}$, the set of products extended by the no purchase option. Let $\phi^k$ be the vector of market share at time $k$, given by $\phi^k_i=\frac{d_i^k}{\sum_{j\in  [n+1]}d_j^k}$, for all $i\in  [n+1]$. Under the assumption that $\{\xi_i\}_{i=1}^{n+1}$ are i.i.d random variables following a Gumbel distribution, and according to standard results for the Multinomial Logit model (see \cite{mcfadden1973conditional} for details), the probability that the $(k+1)$-th consumer purchases product $i$ is given by
\[\pi_i^k=\pi_i(\phi^k, p,q,\beta)=\frac{ (d_i^k)^re^{g_i-\beta_ip_i}}{\sum_{j\in [n+1]} (d_j^k)^re^{g_j-\beta_jp_j}}=\frac{ (\phi_i^k)^re^{g_i-\beta_ip_i}}{\sum_{j\in [n+1]} (\phi_j^k)^re^{g_i-\beta_jp_j}}.\] 
Where $\pi_{n+1}^k+\sum_{i=1}^n\pi_i^k=1$ for all $k\geq0$. We put $\pi^k=(\pi_1^k,\dots, \pi_n^k, \pi_{n+1}^k)$.

A key feature of this model, is that every new customer observes a different network signal in their utilities. This dynamic behaviour could imply that a noisy start could drive the market towards unpredictable outcomes. Figure \ref{fig:F0} depicts the evolution of a market with 4 products (and the no purchase option), where it can be observed that our market model allows to correct misalignments between initial perceptions and the true quality of the products (as long as $0<r<1$).

\begin{figure}[ht]
	\vspace{-3.5em}
	\begin{centering}
		\includegraphics[scale=0.25]{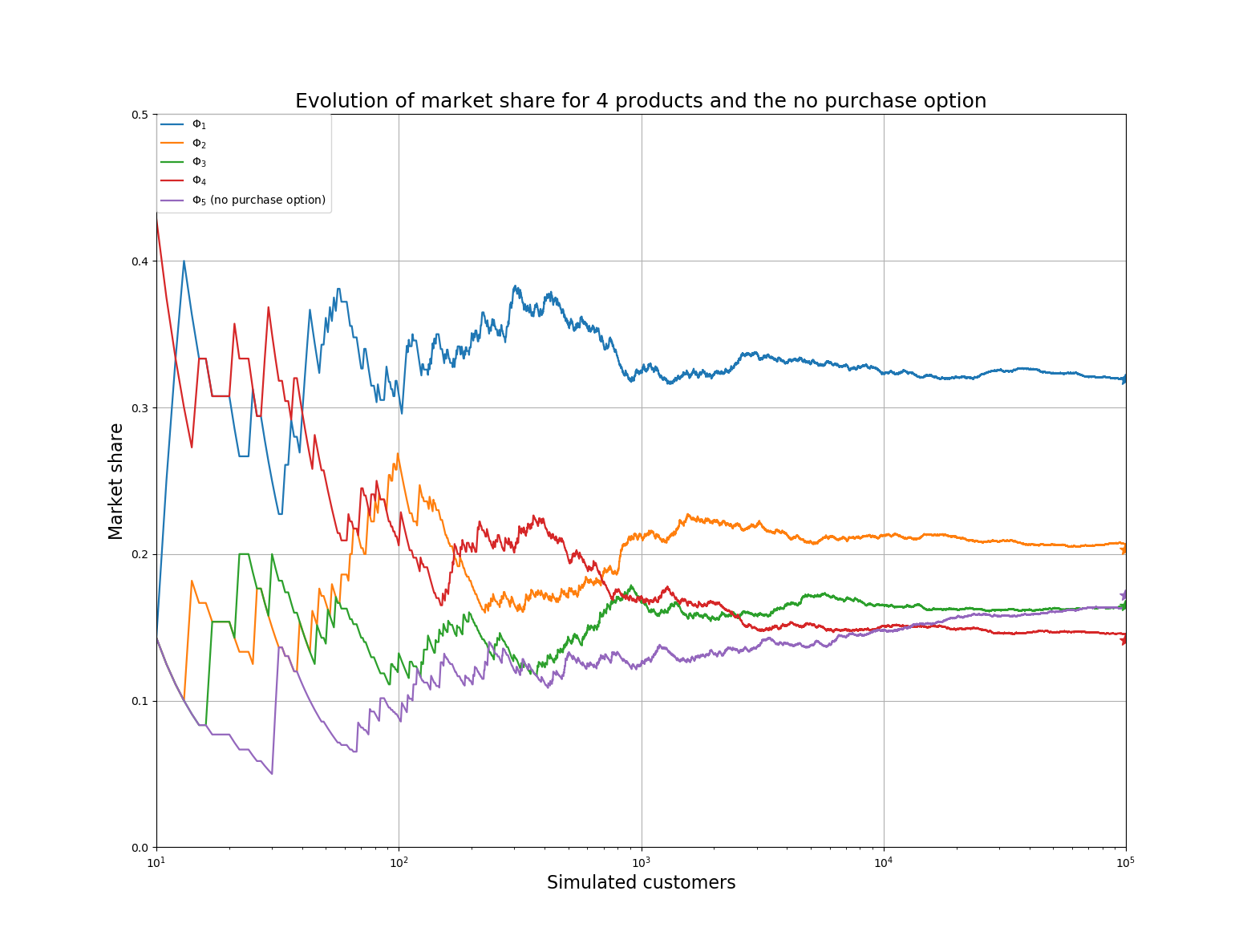}
	\end{centering}
	\centering{}
	\vspace{-4.5mm}
	\caption{\footnotesize Evolution of the market share with 4 products (and the no purchase option), using a network parameter $r=0.5$. The stars represent the respective theoretical convergence points (See Lemma  \ref{MSequilibrium})}
	\label{fig:F0}
	\vspace{-2.5mm}
\end{figure}

%
%
%


The following Lemma establishes an important property that the market share $\phi^k$ satisfies, this result will help to prove that the market eventually stabilise: 

\begin{lemma}
	\label{lem:rma}
	The market share $\phi^k$, satisfies the following recurrence:
	
	\begin{equation}
		\label{rma}\phi^{k+1}=\phi^k+\gamma^{k+1}\left[\pi^k-\phi^k+U^{k+1}\right],
	\end{equation}
	with $\gamma^{k+1}=\mathcal{O}(k^{-1})$ and $U^{k+1}$ a martingale difference noise term (i.e.,  $\EE[U^{k+1}|\phi^t,t\leq k]=0$).
	
\end{lemma}

The recurrence  given by Equation \eqref{rma} is known as a Robbins-Monro Algorithm, and under mild conditions over $\gamma^k, \phi^k,$ and $U^{k}$, the asymptotic behaviour of \eqref{rma} is closely related with the asymptotic behaviour of a continuous deterministic  dynamic given by
\begin{equation}
	\label{rmc}
	\dot \phi^t=\pi(\phi^t)-\phi^t, \quad \phi^t\in \Delta^{n+1}.
\end{equation}
\cite{Ljung77} introduced this idea, commonly called  {\it the ODE Method} for stochastic approximations, ever since it has been extensively studied (e.g., \cite{Duflo97,Kushner03}).
In a recent publication \cite{maldonado2018}  proved that a dynamic like \eqref{rmc} has only one equilibrium with all its coordinates positive, and that under some conditions over the parameters (that are satisfied in our setting), a related Robbins-Monro Algorithm converges almost surely to that equilibrium. Using that result we can  establish  Lemma \ref{MSequilibrium}.

\begin{lemma}
	\label{MSequilibrium}
	For any fixed price $\mathbf{p}=(p_1,...,p_n, p_{n+1})\in\R^n_+\times \{0\} $, fixed parameters $\beta_i, g_i$, and a network effect parameter $r: 0<r<1$, the market share $\phi^k$ converges almost surely to the unique equilibrium $\phi^*=(\phi_i^*)_{i\in  [n+1]}$ given by
	
	\begin{equation}
		\label{MSequil}\phi_i^*=\frac{( e^{g_i-\beta_ip_i})^{1/(1-r)}}{\sum_{j\in [n+1]}( e^{g_j-\beta_jp_j})^{1/(1-r)}}.
	\end{equation}
	
	Furthermore, for every $i\in [n+1]$, $\pi_i(\phi^*)=\phi_i^* $ (fixed point for the probability function $\pi$).
	
\end{lemma}
It is important to notice that according to some results from \cite{maldonado2018}, when $r>1$, $\phi^*$ is an unstable equilibrium, hence the market converges to other equilibria (for example, monopolies for some product) with probability 1.  Establishing pricing policies in those cases using the equilibrium $\phi^*$ as a decision variable does not make sense from a market model perspective (given the associated unpredictability). Therefore, we will focus  only on the cases where $0<r<1$ but allowing $d_i^K$ to grow freely (dynamic market). In some related research (e.g., \cite{cui2016exact} and \cite{wang2016consumer}) the authors do not consider upper bounds on the network parameters, but the market size is fixed (static market). As   \cite{du2016optimal}  point  out,  higher values of those parameters can  lead to suboptimal results (due to multiplicity of equilibria).

In the case that $0<r<1$ we notice that the term $\tau_i:=g_i-\beta_ip_i$ affects directly the expected market share for each product, in particular the product with the highest value of $\tau_i$ gets the largest market share. In this way, if a high intrinsic utility product is too expensive, then the chances of being purchased may decrease, or equivalently lower intrinsic utility products could increase their expected sales after a reduction on their prices. Keeping this into consideration, we define the expected revenue for each seller in terms of the expected market share and the chosen prices.

\begin{definition}
	The expected revenue for seller $i$ is given by
	
	\[w_i=w_i(r,q,p_i,p_{-i})=p_i\phi_i^*=p_i\frac{( e^{g_i-\beta_ip_i})^{1/(1-r)}}{\sum_{j\in [n+1]}( e^{g_j-\beta_jp_j})^{1/(1-r)}}.\]
	
	
\end{definition}

For notational convenience we assume without loss of generality that the intrinsic utilities are non-decreasingly ordered, this is, $g_1\geq g_2\geq \cdots\geq g_n > g_{n+1}=0$, meaning that seller $1$ has the highest intrinsic utility product. In the following two sections we will analyse two types of strategic decisions, that the sellers can follow based on their expected revenues $w_i$.  In Section  \ref{sec:Mono} we analyse the case of a coalition between the sellers where they adopt a monopolistic pricing strategy to maximise the overall expected revenue.  Whereas in Section \ref{sec:compe} we study the case where sellers compete on their prices to maximise their own expected revenues. We will pay special attention to the behaviour of the price and revenue in terms of the network parameter $r$, and when that is relevant we will make explicit the dependance (e.g., $p_i=p_i(r)$).



\section{Monopolistic pricing}
\label{sec:Mono}


We consider in this section a setting where the sellers decide to act  collaboratively. In this context the sellers choose their prices such that they maximise the overall expected revenue defined by
\[R(p)=\sum_{i=1}^nw_i=\sum_{i=1}^np_i\phi_i^*.\]
Thus, we are interested  on finding a price vector $\mathbf{p}^M:=(p_1^M,\dots, p_n^M)$, that we call {\it monopolistic price}, that satisfies

\[\mathbf{p}^M\in \argmax_{p\in \R^n_+}R(p).\]

In Theorem \ref{MonoPrice} we will deduce the conditions that $ \mathbf{p}^M$ must satisfy to maximise $R(p)$, and in the special case of having the same price sensitivities for all the products ( $\beta_i=\beta, \forall i\in\{1,\dots,n\}$), we will provide a closed expression for this price using the  Lambert  $W$ function (see \cite{corless1996lambertw}), where in particular,  for any nonnegative $x$, $W(x)$ is defined as the solution  of the equation 
\[We^W=x.\]
If $x>0$, then $W(x)$ is a positive continuous differentiable function, strictly increasing and concave. The use of the  Lambert  $W$ function spans a wide range of applications, and particularly it has been used  in Economics for pricing on discrete choice models (e.g., \cite{li2011pricing} and \cite{cui2016exact} ).

\begin{theorem}
	\label{MonoPrice}
	The monopolistic price, $\mathbf{p}^M = (p_1^M,p_2^M,\dots,p_n^M)$ that maximises $R(p)$ must satisfy that 
	\begin{equation}
		\label{GenMonPrice}
		\frac{p_i^M}{1-r}-\frac{1}{\beta_i}=\frac{p_k^M}{1-r}-\frac{1}{\beta_k}, \mbox{ for every pair } i,k\in \{1,\dots, n\}
	\end{equation}	
	Furthermore, if the products have the same price sensitivity $\beta_i=\beta, \forall i\in \{1,\dots, n\} $ then all the products have the same price $p^M_i=p^M$ given by 
	\begin{equation}
		\label{monoPrice}
		p^M=\frac{1-r}{\beta}\left[W\left(\frac{\sum_{i=1}^ne^{g_i/(1-r)}}{e}\right)+1\right],
	\end{equation}
	with $W()$ is the Lambert W function.
	
\end{theorem}

Equation \eqref{monoPrice} is in agreement with the results from the classic Multinomial Logit model, where in the case of having the same price sensitivities, leads to the same price for every product. More comparisons can be found in Appendix \ref{sec:comparison}, where we also compare our results to the ones from \cite{du2016optimal}.

The proof for Theorem \ref{MonoPrice} is as follows.

\begin{proof}

	To find the prices that optimise $R(p)$ we compute the gradient of $R$, $\nabla R(p)$, with coordinates $\dfrac{\partial R(p)}{\partial p_k}$ given by
	\vspace{-3.5mm}
	\begin{align*}
		\dfrac{\partial R(p)}{\partial p_k}&=\sum_{i=1, i\neq k}^n\frac{\partial (p_i\phi_i^*)}{\partial p_k}+\frac{\partial (p_k\phi_k^*)}{\partial p_k}\\
		&=\sum_{i=1, i\neq k}^n \frac{\beta_k}{1-r}p_i\phi_i^*\phi_k^*+\phi_k^*-\frac{\beta_k}{1-r}\phi_k^*p_k(1-\phi_k^*)\\
		&=\phi_k^*\left[\frac{\beta_k}{1-r}\left(\sum_{i=1, i\neq k}^n p_i\phi_i^*+p_k\phi_k^*\right)+1-\frac{\beta_k}{1-r}p_k\right]\\
		&=\phi_k^*\left[\frac{\beta_k}{1-r}R(p)+1-\frac{\beta_k}{1-r}p_k\right]
	\end{align*}
	Imposing the first order conditions over $R(p)$ gives us
	
	\[\dfrac{\partial R(p)}{\partial p_k}=0\Leftrightarrow \phi_k^*=0 \vee \frac{R(p)}{1-r}=\frac{p_k}{1-r}-\frac{1}{\beta_k}, \qquad i\in \{1,\dots,n\}, \]
	However $\phi_k^*=0\Leftrightarrow p_k=\infty$, we conclude that for all pairs $i,k$, $ i,k\in \{1,\dots,n\}$, the following equality must hold
	\begin{equation*}
		\label{equiMono}
		\frac{p_k}{1-r}-\frac{1}{\beta_k}=\frac{p_i}{1-r}-\frac{1}{\beta_i},
	\end{equation*}
	which is the desired condition \eqref{GenMonPrice}. Now, defining $z_k=\dfrac{\beta_kp_k}{1-r}$ (that we will call the normalised price for product $k$), Equation \eqref{GenMonPrice} is equivalent to
	
	\begin{equation}
		\label{equiMono2}
		\frac{z_k-1}{\beta_k}=\frac{z_i-1}{\beta_i},\quad \forall i,k \in \in \{1,\dots,n\}.
	\end{equation}
	Equation \eqref{equiMono2}  defines a pairwise relation. On the other hand the prices must also satisfy
	\begin{equation}
		\label{equiMono3}
		\frac{R(p)}{1-r}=\frac{z_k-1}{\beta_k}
	\end{equation}
	Now in the special case when $\beta_i=\beta$ for all the products,  Equation \eqref{equiMono2} implies that all the prices are the same, $p_i=p$ for all $i\in \{1,\dots, n\}$. Replacing this condition into Equation  \eqref{equiMono3} produces the following equivalences
	\begin{align}
		\nonumber \frac{p}{1-r}\sum_{i=1}^n\phi_i^*&=\frac{z-1}{\beta}\\
		\label{priceandMS} z\phi_{n+1}^*&=1\\
		\nonumber\frac{z}{1+e^{-z}\sum_{i=1}^ne^{g_i/(1-r)}}&=1\\
		\label{monoSol} z-1&=e^{-z}\sum_{i=1}^ne^{g_i/(1-r)}\\
		\nonumber (z-1)e^{z-1}&=e^{-1}\sum_{i=1}^ne^{g_i/(1-r)}\\
		\label{normPrice}\Rightarrow z^M&=W(e^{-1}\sum_{i=1}^ne^{g_i/(1-r)})+1.
	\end{align}
	Finally replacing $p^M=\frac{(1-r)z^M}{\beta}$ we have our conclusion.
\end{proof}
We can easily notice that  each coordinate of $\mathbf{p}^M:=(p^M,\dots, p^M)$ is increasing on each value of $g_i$ for all $i\in \{1,\dots, n\}$, this is, higher the intrinsic utility, higher the price. Theorem \ref{MonopMonot} summarises other properties related to the monopolistic price, and the monotonic behaviour of the revenue in terms of the network effect parameter $r$. The proofs for part $(1)$ and $(2)$ are fairly straightforward, and they can be found in the Appendix \ref{sec:proofs}.


%
%
%


\begin{theorem}
	\label{MonopMonot}
	Let all the products have the same price sensitivity $\beta_i=\beta$, and consider a network effect parameter  $r$, $0<r<1$, then the following statements hold true:
	\begin{enumerate}
		\item The market share of the no purchase option, $\phi_{n+1}^*(\mathbf{p}^{M}(r))$, is strictly decreasing in $r$.
		\item The market share of the highest intrinsic utility product, $\phi_1^*(\mathbf{p}^{M}(r))$, is strictly increasing in $r$.
		\item There exists  $r^*, 0<r^*<1$ such that, the overall expected revenue $R(\mathbf{p}^{M}(r))=\sum_{i=1}^np^{M}_i(r)\phi_i^*(\mathbf{p}^{M}(r))$   is strictly  increasing  in $r$ for all $r: r^*\leq r<1$.
	\end{enumerate}
\end{theorem}

\begin{proof}
	\begin{itemize}
		\item[(3)]  We  first notice that if $p^M(r)$ is increasing in some interval $[r^*,r^{**})$, then the conclusion is direct, indeed, since $R(\mathbf{p}^M(r))=p^M(r)(1-\phi_{n+1}^*(\mathbf{p}^M(r)))$, taking the partial derivatives with respect to $r$ gives an expression that it is always positive for $r: r^*<r<1$. We assume then that $p^M(r)$ is decreasing for all $0<r<1$, in particular, we have that if for some $r_1: 0<r_1<1$, $g_i-\beta p^M(r_1)>0$, then for all  $r_2:$ $ r_1<r_2<1$, $g_i-\beta p^M(r_2)> 0$.
		
		Now, we know that the monopolistic price $p^M(r)$ is characterised by Equation \eqref{equiMono3} as follows:
		\[R(\mathbf{p}^M(r))=p^M_i(r)-\frac{1-r}{\beta_i}\]
		and in the special case where all $\beta_i$ are the same, we have
		\[R(\mathbf{p}^M(r))=\frac{1-r}{\beta}\left( \frac{\beta p^M(r)}{1-r}-1\right)=\frac{1-r}{\beta}(z^M(r)-1).\]
		
		Hence, taking the derivative of $R(\mathbf{p}^M(r))$ with respect to $r$, is the same as computing the following
		\begin{align*}
		\frac{\partial R(\mathbf{p}^M(r))}{\partial r}&= \frac{\partial \frac{1-r}{\beta}(z^M(r)-1)}{\partial r}=\frac{1}{\beta}[(1-r)\frac{\partial z^M(r)}{\partial r}-(z^M(r)-1)]
		\end{align*}
		But it can be shown (see Appendix \ref{sec:proofs}) that $\frac{\partial z^M(r)}{\partial r}=\dfrac{1}{(1-r)^2}\sum_{i=1}^n\phi_i^*(p^M(r))g_i$.
		
		Therefore, using again that $\frac{1-r}{\beta}(z^M(r)-1)=R(\mathbf{p}^M(r))=\sum_{i=1}^n\phi_i^*(\mathbf{p}^M(r))p^M(r)$, we have
		\begin{align*}
		\frac{\partial R(\mathbf{p}^M(r))}{\partial r}&=\frac{1}{\beta(1-r)}\left[\sum_{i=1}^n\phi_i^*(\mathbf{p}^M(r))g_i-\sum_{i=1}^n\phi_i^*(\mathbf{p}^M(r))\beta p^M(r)\right]\\
		&=\frac{1}{\beta}\left[\sum_{i=1}^n\phi_i^*(\mathbf{p}^M(r))\left(\frac{g_i-\beta p^M(r)}{1-r}\right)\right]
		\end{align*}
		where $\phi_i^*(\mathbf{p}^M(r))=\dfrac{e^{\frac{g_i-\beta p^M(r)}{1-r}}}{1+\sum_{j=1}^ne^{\frac{g_j-\beta p^M(r)}{1-r}}}$.
		
		For a fixed $r$ we define the following sets:
		\begin{align*}
		N^-(r)&=\{i\in \{1,...,n\}:g_i-\beta p^M(r)\leq 0 \}\\
		N^+(r)&=\{i\in \{1,...,n\}:g_i-\beta p^M(r)> 0 \}
		\end{align*}
		Then we find the following equality  \[\sum_{i=1}^n\phi_i^*(\mathbf{p}^M(r))\left[\frac{g_i-\beta p^M(r)}{1-r}\right]= \sum_{i\in N^-(r)}\frac{e^{\frac{g_i-\beta p^M(r)}{1-r}}\frac{g_i-\beta p^M(r)}{1-r}}{1+\sum_{j=1}^ne^{\frac{g_j-\beta p^M(r)}{1-r}}}+\sum_{i\in N^+(r)}\frac{e^{\frac{g_i-\beta p^M(r)}{1-r}}\frac{g_i-\beta p^M(r)}{1-r}}{1+\sum_{j=1}^ne^{\frac{g_j-\beta p^M(r)}{1-r}}}\]
		
		
		We notice that if $r$ is close enough to 1, then for all $i\in N^-(r)$, $e^{\frac{g_i-\beta p}{1-r}}\frac{g_i-\beta p^M(r)}{1-r}$  is a small negative number, on the other hand for $i\in N^+(r)$, $e^{\frac{g_i-\beta p^M(r)}{1-r}}\frac{g_i-\beta p^M(r)}{1-r}$ is positive and can be arbitrarily large when $r\sim 1$. Necessarily there must exists $r^*$ such that  
		\[\frac{1}{\beta}\left[\sum_{i=1}^n\phi_i^*(\mathbf{p}^M(r^*))\left(\frac{g_i-\beta p^M(r^*)}{1-r^*}\right)\right]>0\]
		
		and as $p^M(r)$ is assumed to be decreasing, we can ensure that there will not be another change of monotony. In conclusion, $R(\mathbf{p}^M(r))$ is a strictly increasing function  when $r^*<r<1$.
	\end{itemize}
\end{proof}

%

The following example shows a small instance where we can see how the prices, market share and revenue are affected under different values of $r$.

\begin{example}
	Consider network parameters  $r\in (0,1)$,  a price sensitivity $\beta_i=\beta=0.1$  and intrinsic utilities given by $(g_1,g_2,g_3,g_4,g_5) = (0.9874,  0.6454,  0.4053, 0.2891, 0.03353 )$. Figure \ref{fig:revemon} depicts the values of monopolistic price ($p^M$),  the market share of the no purchase option  ($\phi_{n+1}^*(\mathbf{p}^M)$) and the highest intrinsic utility product ($\phi_1^*(\mathbf{p}^M)$), and finally the  overall revenue ($R(\mathbf{p}^M)$) as a function of $r$. Figure \ref{fig:revesquare} shows the different expected revenues (the area of the rectangles) for each value of $r$, the total demand is defined as the sum of the expected market shares (not including the no purchase option), and the optimal prices are obtained using Theorem \ref{MonoPrice}. As it can be observed, for lower values  of $r$ the prices are higher but the total demands are lower, the opposite effect is observed when $r$ is close to 1. Figures \ref{fig:F3} and \ref{fig:F4} in  Appendix \ref{sec:exper} also complement these observations.


	\begin{figure}[h]
		\vspace{-1.2em}
		\hspace{-2em}\includegraphics[scale=0.27]{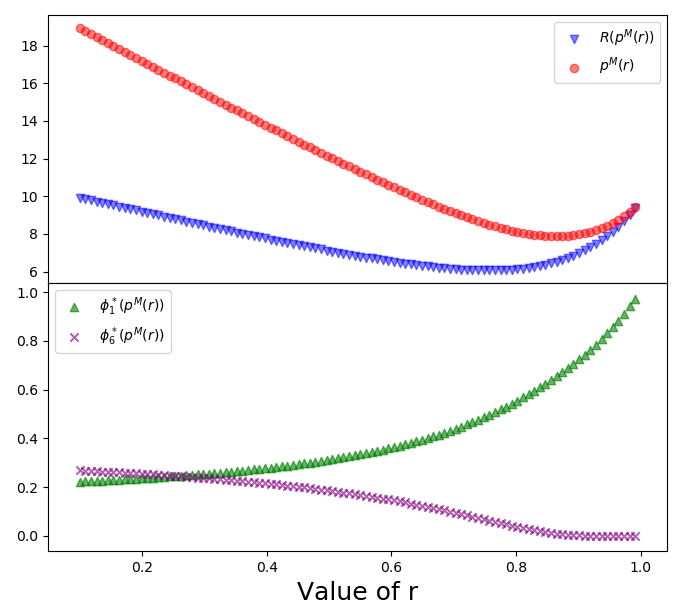}
		\centering{}
		\vspace{-3.5mm}
		\caption{\footnotesize In the top figure, $R(\mathbf{p}^M(r))$ and $p^M(r)$ (blue and red respectively) are displayed for different values of the parameter $r: 0<r<1$. In the bottom figure, the market shares of the highest intrinsic utility product and the no purchase option are displayed (green and purple respectively).}
		\label{fig:revemon}
		\vspace{-3.5mm}
	\end{figure}

	\begin{figure}[h]
		\hspace{-2em}\includegraphics[scale=0.19]{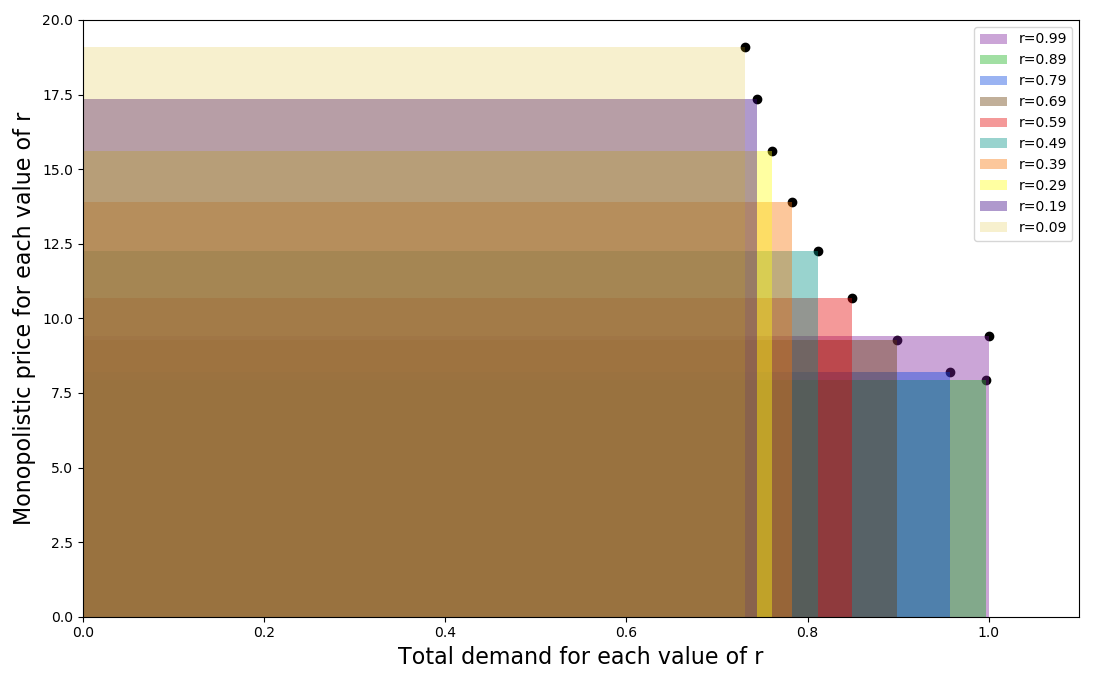}
		\centering{}
		\caption{\footnotesize In the figure, the X axis represents the total demand (scaled up to 1) for the available products, while Y axis contains the prices. The area of each rectangle corresponds to the expected revenue for each value of $r$.}
		\label{fig:revesquare}
	\end{figure}

\end{example}

In Section \ref{sec:compe} we will study the case where the strategic sellers decide to compete to maximise their individual revenues, inducing a {\it price competition game}.

\section{Price competition}
\label{sec:compe}

We consider a complete information  price competition game  $\mathcal{G}=(\{1,\dots, n\},\mathbf{w},S)$, where each player (seller) $i\in\{1,\dots, n\}$ chooses as a strategy a price $p_i$ for his product, in  a common strategy space $S_i=[0,\infty)$. Let $S:=\prod_{i=1}^nS_i=[0,\infty)^n$, and each element $\mathbf{p}\in S$ will be called a strategy profile.

The payoff received by player $i$ after the strategy profile $\mathbf{p}=(p_i, p_{-i})\in S$ is played, is given by $w_i(p)=p_i\phi_i^*(p)$, where $p_{-i}=(p_1,\dots, p_{i-1}, p_{i+1},\dots, p_n)$. We define the joint payoff as $\mathbf{w}=(w_1,\dots, w_n)$. Each player chooses the best response to the other sellers' strategies to maximise their payoff,   hence our objective is to find a maximiser for $\mathbf{w}$. We consider the important notion of Nash Equilibrium in the following definition.

\begin{definition}
	A strategy profile $\mathbf{p}^*=(p_i^*,\dots, p_n^*) \in S$ is a pure Nash Equilibrium (NE) for the game $\mathcal{G}$ if for each player $i$
	\[w_i(p_i^*,p_{-i}^*)\geq w_i(p_i, p_{-i}^*), \forall p_i\in S_i.\]
\end{definition}

The following Theorem shows that there exists a unique pure NE for the game $\mathcal{G}$, given in an implicit form using the Lambert $W$ function.

\begin{theorem}
	\label{NE}
	The price competition game $\mathcal{G}$  has a unique (pure) Nash Equilibrium, $\mathbf{p}^C=(p_1^C,\dots,p_n^C)\in [0,\infty)^n$, with 
	\begin{equation}
		\label{CompPrice}p_i^C=\frac{1-r}{\beta_i}\left[W(\frac{ e^{g_i/(1-r)}}{e+\sum_{j=1, j\neq i}^n e^{\frac{g_j-\beta_jp^C_j}{1-r}+1}})+1\right],\quad \forall i\in\{1,\dots, n\}.
	\end{equation}
	We call $\mathbf{p}^C$,  the {\it competitive price}.
\end{theorem}

\begin{proof}
	We will proceed as follows: first we will show the conditions that the strategy profiles must satisfy in order to be critical points for the vector field $\mathbf{w}=(w_1,\dots, w_n)$; second we will prove that these conditions are also sufficient, meaning that they describe the best response for each player; third we will show that the system of equations that define the best responses has a unique solution; and finally we will conclude.
	
	Indeed, let us consider a vector $\mathbf{p}\in (0,\infty)^n$ and take the first order derivative of $w_i=p_i\phi_i^*(\mathbf{p})$ with respect to $p_i$ for all $i\in  [n+1]$, $i\neq n+1$ (where we are assuming a fixed intrinsic utility vector $\mathbf{g}=(g_1,\dots,g_n)$ and parameters $\beta_i, r$), this is,
	\begin{align*}
		\frac{\partial w_i}{\partial p_i}&=\phi_i^*+p_i\frac{\partial \phi_i^*}{\partial p_i}\\
		&=\phi_i^*+\frac{\beta_i p_i}{1-r}\left[(\phi_i^*)^2-\phi_i^*\right]\\
		&=\phi_i^*\left[1-\frac{\beta_i p_i}{1-r}(1-\phi_i^*)\right]
	\end{align*}
Then $\frac{\partial w_i}{\partial p_i}=0\Leftrightarrow p_i=\frac{1-r}{\beta_i(1-\phi_i^*)}\vee \phi_i^*=0$ . Notice that  $\phi_i^*=0 \Leftrightarrow p_i=\infty$.
The system of equations that define the possible equilibria are given by the conditions
	\begin{equation}
		\label{priceEq}
		\beta_ip_i=\frac{1-r}{1-\phi_i^*} \quad \mbox{ for all } i\in\{1,\dots, n\}.
	\end{equation}
	Calling $z_i:=\frac{\beta_ip_i}{1-r}$, the normalised price, $c_i:=e^{g_i/(1-r)}$  and $M(z):=\sum_{j\in  [n+1]}c_je^{-z_j}$ (with $z_{n+1}=0, c_{n+1}=1$), we  notice that $M(z)$ has the same value for all sellers $i\in  [n+1]$, so in this context can be treated as a constant (for every set of values of prices, $M(z)$ has a fixed value). Equation \eqref{priceEq} can be rewritten as follows: 
	\begin{align}
		\label{sysEq} z_i&=\frac{M(z)}{\sum_{j=0, j\neq i}^nc_je^{-z_j}}, \quad \mbox{ for all } i\in\{1,\dots, n\}\\
		\nonumber\Leftrightarrow (z_i-1)e^{z_i}&=\frac{c_i}{\sum_{j=0, j\neq i}^nc_je^{-z_j}}\\
		\nonumber\Leftrightarrow (z_i-1)e^{z_i-1}&=\frac{c_i}{\sum_{j=0, j\neq i}^nc_je^{-z_j+1}}\\
		\label{Lambert}\Rightarrow z_i-1&=W(\frac{c_i}{\sum_{j=0, j\neq i}^nc_je^{-z_j+1}}), \quad \mbox{ for all } i\in\{1,\dots, n\}
	\end{align}
	$W(\cdot)$ here is the Lambert W function. We have obtained a set of conditions that the critical points of $\mathbf{w}$ must satisfy. Moreover, $z_i$ satisfying Equation \eqref{priceEq} maximises the profit $w_i(z)$, indeed, we consider the second order condition for each function $w_i$ 
	\begin{equation}
		\label{secOrder}\frac{\partial^2 w_i(z)}{\partial p_i^2}=\frac{\partial}{\partial p_i}(\phi_i^*\left[1-\frac{\beta_i p_i}{1-r}(1-\phi_i^*)\right])=\phi_i^*(\phi_i^*-1)\left[z_i-z_i^2+2z_i^2\phi_i^*+\frac{\beta_i}{1-r}\right]
	\end{equation}
	and if $z_i$ satisfies Equation \eqref{priceEq}, then $\phi_i^*=\frac{z_i-1}{z_i}$, and replacing this into \eqref{secOrder} we have 
	\[\frac{\partial^2 w_i(z)}{\partial p_i^2}=-\frac{z_i-1}{z_i^2}\left[z_i^2-z_i+\frac{\beta_i}{1-r}\right]<0,\quad \forall i\in \{1,\dots, n\}.\]
	Then if   $\mathbf{p}^*=(p_1^*,\dots, p_n^*)$ is given by Equation \eqref{priceEq} necessarily, $w(\mathbf{p}^*)\geq w(p_i, p_{-i}^*)$ for all $p_i\in S_i$ for all $i\in N$, this is, $p^*$ is a pure NE. 
	Now, we claim that there is only one solution to the system of equations \eqref{sysEq} (for each set of parameters $\mathbf{g},\boldsymbol{\beta}, r$), defining a unique Nash Equilibrium for the price competition game $\mathcal{G}$. 
	
	Clearly the left hand side of \eqref{sysEq} is increasing in $z_i$, and the right hand side of \eqref{sysEq}, $y_i(z):=\frac{M(z)}{\sum_{j=0, j\neq i}^nc_je^{-z_j}}\in[1,\infty)$ is decreasing for every $z_i$, $0<i\leq n$. Indeed, the denominator of $y_i(z)$ is constant in terms of $z_i$, and the numerator $M(z)=\sum_{j\in  [n+1]}c_je^{-z_j}$ is decreasing in $z_i$  hence there exists a unique intersection of both curves, defining a vector solution $\mathbf{z}^*=(z^*_1,..., z^*_n,z^*_{n+1})=(z^*_1,...,z^*_n,0 )\in (1,\infty)^{n}\times\{0\}$. Finally using that  $z_i^*=\frac{\beta_ip_i^*}{1-r}$ into Equation \eqref{Lambert}, we find that the unique NE, $ \mathbf{p}^C=(p_1^C,\dots, p_n^C)$ is given by 
	
	\begin{equation*}
		p_i^C=\frac{1-r}{\beta_i}\left[W(\frac{ e^{g_i/(1-r)}}{e+\sum_{j=1, j\neq i}^ne^{\frac{g_i-\beta_jp^C_j}{1-r}+1}})+1\right],\quad \forall i:1\leq i\leq n.
	\end{equation*}

\end{proof}

\begin{Rem}
	$p_i^C$ is clearly increasing in terms of its associated intrinsic utility $g_i$ (since $W()$ is increasing), and decreasing in terms of the others products' intrinsic utilities $g_j, j\neq i$. Also for all $i\in [n]$, $p_i^C>\frac{1-r}{\beta_i}$.
\end{Rem}
\vspace{1.5mm}
Generally the competitive price for product $i$, $p_i^C$, depends on the coordinates of the other prices, thus there is no  closed expression for each case. To overcome this issue, we propose the  greedy Algorithm \ref{alg::maxEq}, to compute the value of the  $\mathbf{p}^C$ for any set of parameters $\mathbf{g}, \boldsymbol{\beta}$, and  $r$.
We first consider the following definition: given a vector $x=(\phi_1,...,\phi_n,\phi_{n+1})$, we consider the transformation $\Phi:\R^{n+1}\times\R\times\{1,...,n\} \to \R^{n+1}$ that changes the $i$-th coordinate of $x$ by a given real value $a$, this is, $\Phi_{i,x}(a):=\Phi(x,a,i)=(\phi_1,...,\phi_{i-1},a,\phi_{i+1},...,\phi_n,\phi_{n+1})$

\begin{algorithm}[t]
	\SetAlgoNoLine
	\KwIn{Parameters: $r, c_i=e^{g_i/(1-r)}, i\in\{1,\dots, n\}$, $\epsilon>0$;
		Initial starting point: $z^0\in \R^n_+\times\{0\}$;  }
	\KwOut{A normalised  equilibrium price $z\in \R^n_+\times\{0\}$.}
	 $z\leftarrow z^0$\;
	 \Repeat{$\sqrt{\sum_{i=1}^n|z_i-(W(\frac{c_i}{\sum_{j=0,j\neq i}^{n}c_je^{-z_j+1}})+1)|^2}<\epsilon$}{
		\For{$i\in\{1,\dots, n\}$
		}{
			$z\gets  \Phi_{i,z}(W(\frac{c_i}{\sum_{j=0,j\neq i}^{n}c_je^{-z_j+1}})+1)$;  \quad \quad \quad \quad \quad \quad \quad   [{\bf Update}]  \	
		}
	}
	\caption{Find equilibrium $z^C$ given by Equation \eqref{Lambert}.}
	\label{alg::maxEq}
\end{algorithm}

%

In Lemma \ref{algoeq} we will show that Algorithm \ref{alg::maxEq} always terminates. We will prove it, by exploiting the  fixed point structure on how the  normalised prices are defined.
\begin{lemma}
	\label{algoeq}
	Algorithm \ref{alg::maxEq} is guaranteed to terminate, and its output is the normalised equilibrium price $\mathbf{z}^C=(z_1^C,\dots,z_n^C)$.
	
\end{lemma}

\begin{proof}
	We consider the sequence $(z^k)_{k\in \N}\in \R^{n+1}$ created by each time the Algorithm \ref{alg::maxEq} reaches the step [{\bf Update}],  its coordinates are defined by the recurrence: 
	\begin{equation}
		\label{recurr}z_i^{k+1}=W(\frac{c_i}{\sum_{j=0,j\neq i}^{n}c_je^{-z_j^k+1}})+1 \text{ for all }  i\in\{1,\dots, n\}, \text{ and } k\in \N.
	\end{equation} 
	$(z^k)_{k\in \N}$ is clearly bounded, hence the Bolzano-Weierstrass Theorem ( see for example in \cite{burk2011lebesgue} (Theorem 2.6)) implies that $z^k$ has a convergent subsequence $z^{k_l}$ with $l\in\N$. Let $z^C$ be the limit of $z^{k_l}$. As $W(\cdot)$ is a continuous function for positive arguments. Applying the limit when $l\to \infty$ in both sides of Equation \eqref{recurr},  necessarily $z^C$ must satisfy that for any $ i\in\{1,\dots, n\}$, $z_i^C=W(\frac{c_i}{\sum_{j=0,j\neq i}^{n}c_je^{-z_j^C+1}})+1$, hence Algorithm \ref{alg::maxEq}  terminates when it finds the Equilibrium $z^C$.
\end{proof}

The following Theorem shows the monotonic behaviour of the competitive price $\mathbf{p}^C$ in terms of the network effect parameter $r: 0<r<1$.
\begin{theorem}
	The competitive price $\mathbf{p}^C(r)=(p_1^C(r),\dots,p_n^C(r))\in [0,\infty)^n$ given by Equation \eqref{CompPrice} is decreasing as a function of the network effect parameter $r: 0<r<1$.
\end{theorem}

\begin{proof}
	Imposing the first order conditions over each expected revenue function $w_i(\mathbf{p^C(r)})$, gives us Equation \eqref{priceEq}, which is defined in the following way:
	\[\beta_ip_i^C(r)=\frac{1-r}{1-\phi_i^*(\mathbf{p^C(r)})} \quad \mbox{ for all }  i\in\{1,\dots, n\},\]
	or equivalently:
	\begin{align}
		\nonumber\frac{1-r}{\beta_ip_i^C(r)}&=1-\phi_i^*(\mathbf{p^C(r)}) \quad \mbox{ for all }  i\in\{1,\dots, n\}\\
		\label{priceCompDecre}\Rightarrow  (1-r)\sum_{i=1}^n\frac{1}{\beta_ip_i^C(r)}&=n-1+\phi_{n+1}^*(\mathbf{p^C(r)}).
	\end{align}
	Notice that as $0<\phi_{n+1}^*(\mathbf{p^C(r)})<1$, then Equation \eqref{priceCompDecre}\ implies that
	
	\[\frac{n-1}{1-r}<\sum_{i=1}^n\frac{1}{\beta_ip_i^C(r)}<\frac{n}{1-r}.\]
	Clearly $\frac{n-1}{1-r}$ and $\frac{n}{1-r}$ are increasing in terms of $r$ (and independent of the intrinsic utility parameters), thus necessarily $\sum_{i=1}^n\frac{1}{\beta_ip_i}$ is increasing, which implies that there exists a  product $k\in\{1,\dots,n\}$ such that $p_k^C(r)$ is decreasing. But by definition of the competitive price, we have
	\[p_k^C(r)=\frac{1-r}{\beta_i}\left[W(\frac{ e^{g_k/(1-r)}}{e+\sum_{j=1, j\neq k}^n e^{\frac{g_j-\beta_jp^C_j(r)}{1-r}+1}})+1\right],\quad \forall  k\in\{1,\dots, n\}.\]
	hence if any $p_k^C(r)$ decreases, in order to preserve  the equilibrium, all the other coordinates must decrease as well, which proves the result.
\end{proof}

The following example shows the competitive prices for the case of $3$ products with fixed intrinsic utilities, a fixed value of price sensitivities and 4 different values of network parameters, $r$.

\begin{example}
	\label{PriceComEX} 
	Consider a set of network parameters given by  $r\in\{0.2, 0.4, 0.6, 0.8 \}$, and intrinsic utilities given by $(g_1,g_2,g_3)=(0.993,  0.480, 0.159)$, the competitive  price equilibria $p^C=(p_1^C,p_2^C,p_3^C)$, and the market share and expected revenue for product $i=1$ are given in the following table.
	\begin{center}
		\begin{tabular}{ l | c | c | c || c | c }
			$r$  & $p_1^C$ & $p_2^C$ &$ p_3^C$ & $\phi_1(p)$ & $w_1(p)$\\ \hline 
			0.2  & 11.461  & 9.912   & 9.298   & 0.302       & 3.461\\ 
			0.4  &  9.269  & 7.509   & 6.900   & 0.352       & 3.269\\ 
			0.6  &  7.243  & 5.082   & 4.498   & 0.448       & 3.243\\
			0.8  &  5.612  & 2.581   & 2.121   & 0.644       & 3.613 \\
			\hline
		\end{tabular}
	\end{center}

\end{example}
As we can observe from Examples \ref{PriceComEX}, the highest intrinsic utility product  $i=1$ has in general a decreasing price, and increasing market share, which eventually leads to have a higher revenue when $r=0.8$, a formal explanation of this phenomenon is still an open question. More numerical examples can be observed in the Appendix in Fig. \ref{fig:F1}, where the prices are also compared against the Monopolistic price $p^M$.

In Appendix \ref{sec:homogen} we have included a subcase of the price competition, the {\it homogeneous} case where all the products have the same intrinsic utility (this is, $g_i=g,  i\in\{1,\dots, n\}$). A similar case was studied in \cite{Du2016} so we compare our results against theirs. The main result from that section is that the market share of the no purchase option $\phi_{n+1}^*(r)$ is a decreasing function of $r$ (which we interpret as that in the presence of stronger network effects, people tend to purchase more). This result seems to be true also in the general case ($g_i$ different), but we only have observed this empirically (see Figure \eqref{fig:F5} in Appendix \ref{sec:exper}). 

In the following section we will compare the two different pricing strategies, including also the consumer's perspective.


\section{Monopolistic vs Competitive}
\label{sec:Compare}



In this section we will compare the different pricing schemes where network effects are present, in absolute terms (which prices are higher) and in relative terms from the consumer's perspective. We assume from now on, that the products have the same price sensitivities $\beta_i=\beta$ for all $ i\in\{1,\dots, n\}$. 
The following theorem recovers the intuitive result that the monopolistic price is higher than the competitive one.

\begin{theorem}
	\label{ComparePrice}
	For any set of parameters $g_i, i\in\{1,\dots, n\}$, $0\leq r<1$ and $\beta>0$, the monopolistic price $p^M$ is higher than the competitive price, $p_i^C$ for all $i\in\{1,\dots,n\}$.
\end{theorem}

\begin{proof}
	We know that according to Equations \ref{monoPrice}, and \ref{CompPrice}, $p^M$ and $p_i^C$ are given respectively by
	\begin{align*}
		p^M&=\frac{1-r}{\beta}\left[W\left(\frac{\sum_{i=1}^ne^{g_i/(1-r)}}{e}\right)+1\right]\\
		p_i^C&=\frac{1-r}{\beta}\left[W\left(\frac{e^{g_i/(1-r)}}{e+\sum_{j\neq i}^ne^{1+\frac{g_j-\beta p_j^C}{1-r}}}\right)+1\right]
	\end{align*}
	Their respective vector forms are given by: $\mathbf{p}^M=(p^M,\dots, p^M)$ and  $\mathbf{p}^C=(p_1^C,\dots, p_n^C)$. Comparing both expressions we have that for any set of parameters  $g_i, i\in\{1,\dots, n\}$, $0\leq r<1$ and $\beta>0$ and for any product $i\in \{1,\dots,n\}$
	\[p^M\geq p_i^C\Leftrightarrow \frac{1}{e}\left[\sum_{i=1}^ne^{g_i/(1-r)}-\frac{e^{g_i/(1-r)}}{1+\sum_{j\neq i}^ne^{\frac{g_j-\beta p_j^C}{1-r}}}\right]\geq 0,\]
	But,
	\[\left[\sum_{i=1}^ne^{g_i/(1-r)}-\frac{e^{g_i/(1-r)}}{1+\sum_{j\neq i}^ne^{\frac{g_j-\beta p_j^C}{1-r}}}\right]=\underbrace{e^{g_i/(1-r)}\left(1-\frac{1}{1+\sum_{j\neq i}^ne^{\frac{g_j-\beta p_j^C}{1-r}}}\right)}_{:=A}+\underbrace{\sum_{j\neq i}^ne^{g_j/(1-r)}}_{:=B}\]
	
	Clearly $B>0$ and  since $e^{x}>0$ for any value of $x$, then for all $i\in\{1,\dots, n\}$ , $A>0$.
	Consequently $p^M\geq p_i^C$ as desired.
\end{proof}

The following theorem shows that for any product, the consumer's expected utility obtained from purchasing it,  is higher when the competitive price is used instead of the monopolistic price. This result is trivial when there is no network effects ($r=0$) since the utility is a decreasing function of the price, however if we include the non-linear effect of past purchases the result is not necessarily obvious (given the non linear dependency of the price in the market share).

\begin{theorem}
	\label{CompareCustUtil}
	For any product $ i\in\{1,\dots, n\}$,  in the long run, the expected utility perceived by a customer after purchasing product $i$ when the competitive price is used,  is higher than the case when  the monopolistic price is used.
\end{theorem}

\begin{proof} We want to prove that asymptotically $ u_i^k(\mathbf{p}^C)-u_i^k(\mathbf{p}^M)$ is strictly positive, with $u_i^k(\mathbf{p})$ given by Equation \eqref{customerUt}. We know that by to Lemma \ref{MSequilibrium}, $\frac{d_i^k(\mathbf{p})}{k}\xrightarrow[a.s.]{}\phi_i^*(\mathbf{p})$, then

	\begin{align*}
		u^k_i(\mathbf{p}^C)-u^k_i(\mathbf{p}^M)\xrightarrow[a.s.]{}& r[\log(\phi_i^*(\mathbf{p}^C))-\log(\phi_i^*(\mathbf{p}^M))]-\beta(p_i^C-p^M)\\
		&=r\log\left[\frac{\phi_i^*(\mathbf{p}^C)}{\phi_i^*(\mathbf{p}^M)}\right]+\beta(p^M-p_i^C)
	\end{align*} 
	
	According to Theorem \ref{ComparePrice}, we know that $\beta(p^M-p_i^C)>0$, on the other hand we have $\phi_i^*(\mathbf{p})=\dfrac{e^{g_i/(1-r)}}{e^{\frac{\beta p_i}{1-r}}+\sum_je^{g_j/(1-r)}e^{\frac{\beta (p_i-p_j)}{1-r}}}$, which is clearly decreasing in terms of $p_i$, then as $p^M>p_i^C$ for all $i\in \{1,..,n\}$, necessarily $r\log\left[\frac{\phi_i^*(\mathbf{p}^C)}{\phi_i^*(\mathbf{p}^M)}\right]>0$ for all $i$, meaning that $r\log\left[\frac{\phi_i^*(\mathbf{p}^C)}{\phi_i^*(\mathbf{p}^M)}\right]+\beta(p^M-p_i^C)>0$ as desired. 
\end{proof}

The structure of the market share in the equilibrium (Equation\eqref{MSequil}) implies that the highest market share would be assigned to the product with highest value of $\dfrac{g_i-\beta_ip_i}{1-r}$ which at least for the competitive price, $\mathbf{p}^C$ is a increasing function of $r$, meaning that in general, the consumer's utility  associated to the product with highest intrinsic utility ($i=1$),  increases as $r$ approaches to 1. The following example depicts this effect.

\begin{example}
	\label{utilIncre}
	Consider a large enough amount of customers such that for any product $i\in\{1,\dots,n\}$, $\frac{d_i^k}{k+1}=\phi_i^*(\mathbf{p})$, where $\phi_i^*(\mathbf{p})$ is the market share in the equilibrium (see Equation\eqref{MSequil}) under a price $p$ (competitive Nash Equilibrium and/or monopolistic price). Customer $k+1$ will then choose strategically a product $j=j(q,r,\beta, p, \xi)$ that maximises his expected utility of purchasing any product (or he will choose the no purchase option), this is, using formula \eqref{customerUt}, we have
	
	\[j\in \argmax_{0\leq i\leq n}\EE[g_i+r\ln(d_i^k)-\beta p_i+\xi_i]\]
	
	Where $\xi_i$, $i\in[n+1]$ were chosen to be i.i.d random variables following a Gumbel distribution, in particular we have that $\EE[\xi_i-\xi_j]=0$ for all pairs $i,j\in\{1,\dots,n,n+1\}$. Let $v_j(\mathbf{p}):=u_j(\mathbf{p})-\xi_j$ and consider the following parameters: $\beta_i=\beta=0.1$, $k=10K$, and intrinsic utilities given by $(g_1,g_2,g_3)=(0.993,  0.480, 0.159)$
	
	The following table summarises how the expected utilities, $\EE[v_j]$, behave under different values of $r$. The second and third  columns show which product, $j$ is the one  that maximises the expected utility, under the competitive and monopolistic pricing ( $j^C$  and $j^M$ respectively). The fourth and fifth column show the respective competitive and monopolistic prices for those products. Finally, the last two columns show the expected values of $v_{j^C}$ and $v_{j^M}$ respectively.
	
	\begin{center}
		\begin{tabular}{c | c | c | c |  c | c | c  }
			$r$ & $j^C$ & $j^M$ & $p_{j^C}^C$ & $p^M$ & $\EE[v_{j^C}(\mathbf{p}^C)]$ & $\EE[v_{j^M}(\mathbf{p}^M)]$ \\ \hline 
			0.2 & 1 & 1 &11.461& 15.498 & 2.831 & 2.395 \\ 
			0.4 & 1 & 1  & 9.269 & 12.523 & 6.097 &5.721 \\ 
			0.6 & 1 & 1 & 7.243& 9.798 & 9.458 & 9.167 \\ 
			0.8 & 1 & 1 & 5.613& 7.934 & 12.974 & 12.791 \\ 
			
			\hline
		\end{tabular}
	\end{center}

\end{example}

Where in each one of the cases, the highest intrinsic utility product ($j=1$) is the one with the largest chances of being chosen.

\section{Conclusions and open problems}
\label{sec:conclusion}

In this work we have designed a model for consumer choice, based on a MNL model with non-linear network effects.
We studied a multi-seller pricing problem where sellers can collaborate or compete, finding in each case a unique equilibrium price (monopolistic price and Nash Equilibrium respectively). We also studied the monotonic behaviour of the market shares, prices and revenues in terms of the network parameter $r$ (both theoretically and numerically). We finally compared both pricing strategies from the consumer's perspective, recovering for our model some well known results from the traditional MNL, such as that the monopolistic price is higher than the competitive one, and that the utility perceived by the consumers is higher when the competitive price is used. We also analysed numerically how increasing the network parameter $r$ generates higher utilities for the consumer.

Some interesting questions remain open, for example, numerically we detected that  the revenue for the highest quality product $w_1$ in the competitive case seems to decrease and then increase when  the network parameter $r$ approaches to 1, then {\it is there a critical value $\hat{r}$ such that for all $r: 1<r<\hat{r}$,  $w_1$ is increasing}? Numerically, we also have observed that the expected utility for the costumers seem to increase with the value of $r$, but we do not have a proof for that phenomenon. Answering those kind of questions would help to find the best value of $r$ such that both consumers and sellers are benefited from network effects.

\bibliographystyle{ACM-Reference-Format}
\bibliography{aux_files/references}

\newpage
\appendix
\section{MODEL COMPARISON}
\label{sec:comparison}
In the context of monopolistic pricing, many research has been done using variations of the Multinomial Logit model, particularly pricing under the standard  MNL has some very well known properties (e.g., for products with the same intrinsic quality, the optimal (monopolistic) price is the same for all the products ). In general, the functional form of the consumers' utilities lead to different probability functions that drive the behaviour of the model. It is important  then, to compare $\mathbf{p}^M$ from our model against the monopolistic price obtained with other models (classic MNL, and the model proposed by \cite{du2016optimal}). To do that we proceed to characterise the two models  we will be comparing against (under our notation), using their probabilities. 
\begin{definition}
	The probability $\pi_i^C$ of choosing product  $i\in \{1,\dots, n\}$   for the classic MNL model (without network effects) is given by 
	\[\pi_i^C=\frac{exp(g_i-\beta_ip_i)}{1+\sum_{j=1}^nexp(g_j-\beta_jp_j)}\]
	where $g_i$,$\beta_i$ and $p_i$ are defined as before. We put $\pi^C=(\pi_1^C,\dots,\pi_{n+1}^C)$.
\end{definition}

\begin{definition}
	For the MNL model with network effects defined in  \cite{du2016optimal}, the probability $\pi_i^D$ of choosing product  $i\in \{1,\dots, n\}$ is given by 
	\[\pi_i^D=\frac{exp(g_i-\beta_ip_i+\alpha_i\phi_i)}{1+\sum_{j=1}^nexp(g_j-\beta_jp_j+\alpha_j\phi_j)}\]
	where $\alpha_i$ is the network sensitivity of product $i$, and $\phi_i$ its market share. We put $\pi^D=(\pi_1^D,\dots,\pi_{n+1}^D)$.
\end{definition}
The following table summarises some of the comparisons we obtain when we consider the different probability models. The first column of the table contains the settings where we will be making the comparisons, the second column contains the conclusions given by our probability distribution $\pi=(\pi_1,\dots, \pi_{n+1})$, where $\pi_i=\dfrac{\phi_i^rexp(g_i-\beta_ip_i)}{\sum_{j\in[n+1]}\phi_j^rexp(g_j-\beta_jp_j)}$. The third and fourth columns contain the results when $ \pi^C, \pi^D$ are  used, respectively.

\begin{center}
	\begin{tabular}{ p{1.5cm} | p{4cm} | p{3.5cm} | p{4cm}  }
		Setting & $\pi$ & $\pi^C$ & $\pi^D$ \\ \hline 
		$\beta_i=\beta$  for all $i\in[n]$& Unique optimal price is to assign the same price to every product (uniform price).   & Uniform price. & No explicit form for the optimal price. \\  
		& Market shares are increasing on the intrinsic utility of the products. &Market shares are increasing on the intrinsic utility of the products. &  Since, in their model, the following expression  must be constant  $2\alpha\phi_i-\log(\phi_i)+g_i$, then if for some $i$, $\phi_i>\frac{1}{2\alpha}$,  an increment on its intrinsic utility, would lead to a decrement of its market share (Lemma 4.1 in \cite{du2016optimal}). \\ \hline
		$\beta_i=\beta$, $g_i=g$ for all $i\in[n]$&  Uniform price & Uniform price & Uniform price if $\alpha<\hat{\alpha}$, for some $\hat{\alpha}$. Otherwise, uniform price for $n-1$ products, and one product with a lower price (Theorem 3.2 in \cite{du2016optimal}.) \\
		& Uniform market share & Uniform market share & Uniform market share if $\alpha<\hat{\alpha}$, otherwise, the cheapest product has a larger market share (Theorem 3.1 in \cite{du2016optimal}).\\ \hline

		\hline
		
		\hline
	\end{tabular}
\end{center}

Finally, it is worth mentioning that in our model, even if the price sensitivities are different, according to Equation \eqref{GenMonPrice}, if $r\to 1$, then for all $i,k\in \{1,...,n\}$, $p_i^M=p_k^M$. However, in that case the highest intrinsic utility product gets a market share close to 1, while the rest of the products have a negligible share (a monopoly for the highest intrinsic utility product).

\section{SUBCASES: PRICE COMPETITION HOMOGENEOUS CASE}
\label{sec:homogen}
In this section we present a simplification of the general case of price competition, where every product presents the same intrinsic utility. This case will allow us to study, from a theoretical point of view, the behaviour of the prices as a function of the network parameter $r$.
We assume in this section that the values $g_i=g$ for all $i: 1\leq i\leq n$, and we define  for notational convenience $\hat{c}=e^{g/(1-r)}$ . The following corollary is a direct consequence of Theorem \ref{NE} for the case where all products have the same intrinsic utility. 

\begin{corollary}
	If all the products have the same intrinsic utility, $g_i=g$ for all $i\in N$, then the competitive price for the homogeneous case, $p^{CH}=(p_1^{CH},\dots,p_n^{CH})$ is the unique pure NE for the game  $\mathcal{G}$ , and its coordinates are given by
	\begin{equation}
	\label{homogPrice}
	p_i^{CH}=\frac{1-r}{\beta_i}\left[W(\frac{\hat{c}}{e+\hat{c}(n-1)e^{1-\frac{\beta_i p_i^{CH}}{1-r}}})+1\right],\quad \forall 1\leq i\leq n.
	\end{equation}
	
\end{corollary}

\begin{proof}
	Thanks to Theorem \ref{NE}, we know that the coordinates of the unique NE for the price competition are given by Equation \eqref{CompPrice}. Now in the particular case where all the products have the same intrinsic utility, Equation \eqref{Lambert} gets reduced to
	\begin{equation*}
	z_i-1=W(\frac{\hat{c}}{e+\hat{c}\sum_{j=1,j\neq i}e^{-z_j+1}}),\quad \forall i: 1\leq i\leq n
	\end{equation*}
	which is completely symmetric for each $z_i$, therefore for all $1\leq i\leq n$, it must hold $z_i=z$ for some $z>1$. Consequently the previous Equation  is equivalent to
	\begin{align}
	\label{symm}z-1&=W(\frac{\hat{c}}{e+\hat{c}(n-1)e^{-z+1}})\\
	\nonumber\Rightarrow p_i&=\frac{1-r}{\beta_i}\left[W(\frac{\hat{c}}{e+\hat{c}(n-1)e^{1-\frac{\beta_i p_i}{1-r}}})+1\right].
	\end{align}
	
\end{proof}

\begin{Rem}
	Even when the solution for $z_i$ is given by a fixed value $z_i=z$ for all $1\leq i\leq n$, the prices $p_i$ can be different, due to the sensitivity parameter $\beta_i$. This phenomenon has also been studied in  \cite{ezra2017pricing} where the authors analyse the problem of pricing identical items, that eventually leads to different prices depending on consumption patterns.
\end{Rem}

The following Theorem states similar properties to Theorem \ref{MonopMonot} but now for the case of the competitive price. We are able to prove some monotonic behaviour of the normalised price, the products' market share, and the market share of the no purchase option when the competitive homogeneous price is used. However, similar properties seem to hold also for the general case (see Example \ref{MSCompEX}).

\begin{theorem}
	\label{CompHMonot}
	Under the assumption of homogeneity in the intrinsic utilities (i.e. $g_i=g$ for all $i\in \{1,...,n\}$), if we consider a network effect parameter  $r$, $0<r<1$, then the following statements hold true:
	\begin{enumerate}
		\item The normalised competition price $z^{CH}$ is increasing in terms of $r$.
		\item Every product has the same market share $\phi_i^*(p^{CH})=\phi^*(p^{CH})$ which is increasing in $r$.
		\item The market share for the no purchase option, $\phi_{n+1}^*(p^{CH})$,  is decreasing as a function of $r$.
	\end{enumerate}
\end{theorem}

\begin{proof}
	\begin{enumerate}
		\item We prove first that our normalised price $z$  is increasing in terms of $r$. Indeed, we notice that by definition of Lambert $W$ function, Equation \eqref{symm} is equivalent to 
		\[(z^{CH}-1)e^{z^{CH}}+z^{CH}\hat{c}(n-1)=n\hat{c}.\]
		
		Taking the derivative with respect to $r$ in both sides of the Equation, we find the following:
		\begin{align}
		\nonumber z^{CH}\frac{\partial z^{CH}}{\partial r}e^{z^{CH}}+\frac{\partial z^{CH}}{\partial r}\hat{c}(n-1)+z^{CH}(n-1)\frac{\partial \hat{c}}{\partial r}&=n\frac{\partial \hat{c}}{\partial r}\\
		\label{decreHomo}\frac{\partial z^{CH}}{\partial r}[z^{CH}e^{z^{CH}}+(n-1)\hat{c}]&=\frac{\hat{c}g[n-(n-1)z^{CH}]}{(1-r)^2}
		\end{align}
		On the other hand, according to Equation \eqref{symm} we see that
		\[z^{CH}-1=\frac{\hat{c}}{e^{z^{CH}}+\hat{c}(n-1)}<\frac{\hat{c}}{\hat{c}(n-1)}=\frac{1}{n-1},\]
		and then $z^{CH}<\dfrac{n}{n-1}$. Using this into Equation \eqref{decreHomo} we obtain that $\dfrac{\partial z^{CH}}{\partial r}>0$, where $z$ is a increasing function of $r$.

		\item We notice that each market share in the equilibrium is given by $\phi_i^*(p^{CH})=\dfrac{\hat{c}e^{-z_i^{CH}}}{1+\hat{c}ne^{-z_i^{CH}}}=\dfrac{\hat{c}}{e^{z^{CH}}+\hat{c}n}:=\phi^*(p^{CH})$ which is independent of $i$, since all the normalised prices $z_i$ are the same. Taking the derivative of $x$ with respect to r give us the following equalities: 
		\begin{align*}
		\frac{\partial \phi^*( p^{CH})}{\partial r}&=\frac{\partial }{\partial r}\left[\frac{\hat{c}}{e^{z^{CH}}+ n\hat{c}}\right]\\
		&=\frac{\hat{c} \phi^*( p^{CH})}{(1-r)^2}\left[\frac{z^{CH}e^{2z^{CH}}+(z^{CH}(n-1)-1)\hat{c}e^{z^{CH}}}{(z^{CH}e^{z^{CH}}+(n-1)\hat{c})(e^{z^{CH}}+ n\hat{c})}\right]>0.
		\end{align*}
		Hence all the market share in the equilibrium are increasing in terms of the network parameter $r$.

		\item Since $\phi_{n+1}^*( p^{CH})=1-\sum_{i=1}^n\phi_i^*( p^{CH})=1-n\phi^*( p^{CH})$, and $\phi^*(\cdot)$ is increasing in $r$, necessarily $\phi_{n+1}^*( p^{CH})$ must be decreasing.
		
		
	\end{enumerate}
	
\end{proof}

\begin{Rem}
	Numerical simulations have shown us that similar conclusions from Theorem \ref{CompHMonot} in the points $1.$ and $3.$ (normalised price increasing and market share for no purchase option decreasing) seem to hold for the general competition case. However, we only have been able to observe it empirically (see for example Figures \ref{fig:F1} to \ref{fig:F5} in the Appendix \ref{sec:exper}). 
\end{Rem}
The following example shows  that Theorem \ref{CompHMonot} part (2). does not necessarily hold  when the intrinsic utilities are different, where there are some products whose market shares decrease in terms of $r$. We also can observe that the market share for the highest intrinsic utility products seems to be increasing.

\begin{example}
	\label{MSCompEX}
	Consider a set of network parameters given by $r\in\{0.2, 0.4, 0.6, 0.8 \}$,  intrinsic utilities given by $(g_1,g_2,g_3)=(0.993,  0.480,  0.159)$, and a price sensitivity $\beta_i=\beta=0.1$, the market share for each product and their respective expected revenue for each $r$ are given in the following table.
	\begin{center}
		\begin{tabular}{ l | c | c | c | c | c | c}
			$r$ & $\phi_1(p)$ & $\phi_2(p)$ &  $\phi_3(p)$ & $w_1(p)$ & $w_2(p)$ &  $w_3(p)$  \\ \hline 
			0.2 & 0.302 & 0.193 & 0.134 & 3.461 & 1.912 & 1.298 \\ 
			0.4 & 0.352 & 0.201 & 0.130 & 3.269 & 1.509 & 0.900 \\ 
			0.6 & 0.448 & 0.213 & 0.111 & 3.243 & 1.082 & 0.498 \\ 
			0.8 & 0.644 & 0.225 & 0.057 & 3.613 & 0.581 & 0.121 \\ 
			\hline
		\end{tabular}
	\end{center}
\end{example}

\section{PROOFS}
\label{sec:proofs}

\begin{proof}[Proof of Lemma \ref{lem:rma}]
	Following the idea of \cite{maldonado2018}, consider that in each time step $k$ (arrival of $k$-th consumer) either a product $i\in\{1,...,n\}$ is purchased, or no product is purchased ($i=n+1$), then, defining $D^k:=\sum_{j\in [n+1]}d_j^k=\sum_{j\in [n+1]}\sum_{t=1}^kd_j^t=k$,  we have that $\phi^k=D^k\dfrac{\phi^k}{D^k}\Rightarrow \phi^{k+1}=\frac{D^k\phi^k+e^{k+1}}{D^{k+1}},$ with $e^{k+1}$ a random ($n+1$- dimensional) variable with coordinates $(e^{k+1})_i=1$ if product $i\in \{1,...,n\}$ has been purchased at time $k+1$,  $(e^{k+1})_{j\neq i}=0$; and $(e^{k+1})_{n+1}=1$ if no product is purchased by the consumer $k+1$. 
	
	Hence, clearly  $\EE[e^{k+1}|\phi^t,t\leq k]=\pi^k$, and  considering $\gamma^{k+1}:=\dfrac{1}{D^{k+1}}=\dfrac{1}{k+1}$, and $U^{k+1}:=e^{k+1}-\EE[e^{k+1}|\phi^t,t\leq k]$, we get the desired recurrence 
	\begin{equation*}
	\phi^{k+1}=\phi^k+\gamma^{k+1}\left[\pi^k-\phi^k+U^{k+1}\right].
	\end{equation*}

\end{proof}

\begin{proof}[Proof of Theorem \ref{MonopMonot} parts (1) and (2)]
	Let $p^{M}(r)$ be the monopolistic price for each product, given by Theorem \ref{MonoPrice}, and consider the normalised price $z^{M}(r)=\frac{\beta p^{M}(r)}{1-r}$.
	\begin{enumerate}
		\item  We know that according to Equation \eqref{priceandMS}, the market share for the no purchase option, $\phi_{n+1}^*(\mathbf{p}^M(r))$, must satisfy $z^{M}(r) \phi_{n+1}^*( \mathbf{p}^{M}(r))=1$. Clearly since $W()$ is an increasing function, Equation \eqref{normPrice} implies that  $z^{M}(r)$ is strictly increasing in terms of $r$, hence  $\phi_{n+1}^*( \mathbf{p}^{M})(r)$ must be strictly decreasing as a function of $r$.
		
		\item We first compute the derivative of $z^M(r)$ with respect to $r$, indeed we use Equation \eqref{monoSol} to obtain the $\frac{\partial z^{M}(r)}{\partial r}$ as follows
		\begin{align*}
		\frac{\partial z^{M}(r)}{\partial r}&=-\frac{\partial z^{M}(r)}{\partial r}e^{-z^{M}(r)}\sum_{i=1}^ne^{g_i/(1-r)}+e^{-z^{M}(r)}\sum_{i=1}^n\frac{g_ie^{g_i/(1-r)}}{(1-r)^2}\\
		\Rightarrow \frac{\partial z^{M}(r)}{\partial r}&=\frac{1}{(1-r)^2}\frac{\sum_{i=1}^ne^{-z^{M}(r)}e^{g_i/(1-r)}g_i}{1+e^{-z^{M}(r)}\sum_{i=1}^ne^{g_i/(1-r)}}=\frac{1}{(1-r)^2}\sum_{i=1}^n\phi_i^*(\mathbf{p}^{M}(r))g_i.
		\end{align*}

		Now we consider the market share for the highest intrinsic utility product, $\phi_1^*(\mathbf{p}^M(r))$, and we take its first derivative with respect to $r$: 
		\begin{align*}
		\frac{\partial \phi_1^*( \mathbf{p}^{M}(r))}{\partial r}&=\frac{\frac{\partial e^{g_1/(1-r)}}{\partial r}e^{-z^{M}(r)}+\frac{\partial z^{M}(r)}{\partial r}e^{-z^{M}(r)}e^{g_1/(1-r)}}{1+e^{-z^{M}(r)}\sum_{i=1}^ne^{g_i/(1-r)}}-\\
		&\frac{e^{-z^{M}(r)}e^{g_1/(1-r)}}{(1+e^{-z^{M}(r)}\sum_{i=1}^ne^{g_i/(1-r)})^2}\left(e^{-z^{M}(r)}\sum_{i=1}^n\frac{\partial e^{g_i/(1-r)}}{\partial r}-\frac{\partial z^{M}(r)}{\partial r}e^{-z^{M}(r)}\sum_{i=1}^ne^{g_i/(1-r)}\right)\\
		&=\frac{ \phi_1^*( p^{M}(r))}{(1-r)^{2}}\left[g_1-\sum_{j=1}^n\phi_j^*(\mathbf{p}^{M}(r))g_j     +[1-\sum_{j=1}^n\phi_j^*(\mathbf{p}^{M}(r))]\sum_{j=1}^n\phi_j^*(\mathbf{p}^{M}(r))g_j\right],
		\end{align*}
		the only term that can be negative in the last equality is $g_1-\sum_{j=1}^n\phi_j^*(\mathbf{p}^{M})g_j $, but as $\sum_{j\in [n+1]}\phi_j^*=1$, then $g_1=\sum_{j\in [n+1]}\phi_j^*g_1=\sum_{j=1}^n\phi_j^*g_1+\phi_{n+1}^*g_1$, and  since $g_1\geq g_j$, for all $1\leq j\leq n$ we have \[g_1-\sum_{j=1}^n\phi_j^*(\mathbf{p}^{M})g_j =\sum_{j=1}^n(g_1-g_j)\phi_j^*(\mathbf{p}^{M})+\phi_{n+1}^*(\mathbf{p}^{M})g_1>0.\]
		In conclusion $\phi_1^*(p^M)$, the market share for the highest intrinsic utility product is strictly increasing in terms of $r$.

	\end{enumerate}

\end{proof}

\section{ADDITIONAL EXPERIMENTS}

\label{sec:exper}
We present here some extra experimental results depicting the different behaviour of both pricing schemes (competitive price against monopolistic price). We use the following parameters:
$g=(g_1,....,g_5)=(0.850, 0.733, 0.416, 0.256, 0.139 )$, $\beta_i=\beta=0.1$.

\begin{figure}[ht]
	\vspace{-3.5mm}
	\begin{centering}
		\includegraphics[scale=0.72]{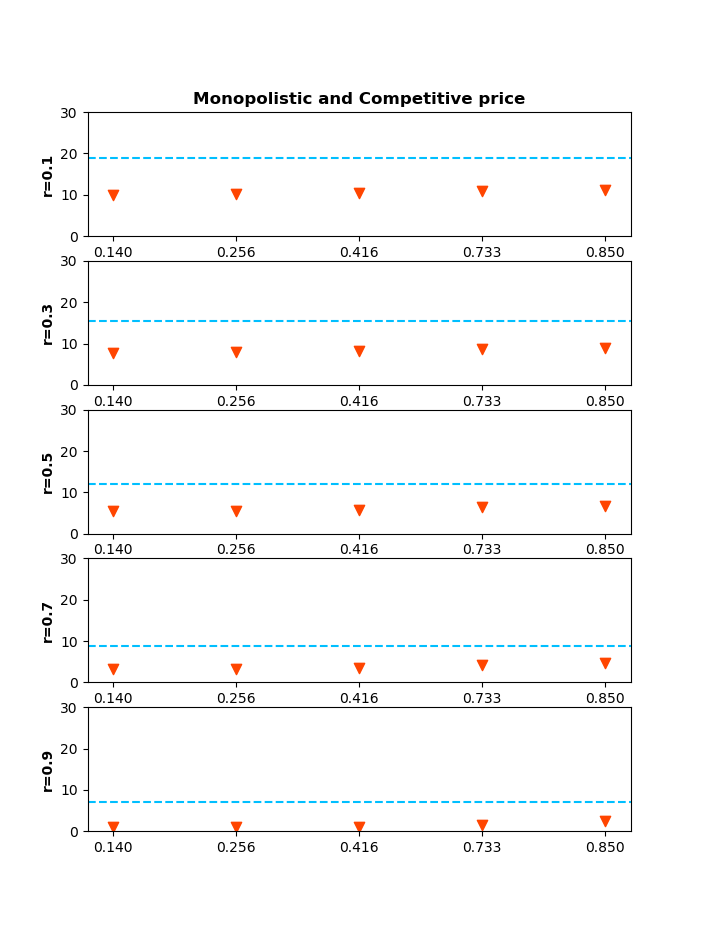}
	\end{centering}
	\centering{}
	\vspace{-3.5mm}
	\caption{\small  Comparison of prices, competition versus monopoly respect to the products'  intrinsic utility (X axis). The red triangles are the competitive prices (NE) for each product, and the blue dotted line is the monopolistic price  }
	\label{fig:F1}
	\vspace{-2.5mm}
\end{figure}

In Fig. \ref{fig:F1} we observe how the prices (both $p^C, P^M$) decrease as a function of $r$. And that in the competitive case the prices seem to be increasing in the value of their intrinsic utilities ($g_i$).

\begin{figure}[ht]
	\includegraphics[scale=0.5]{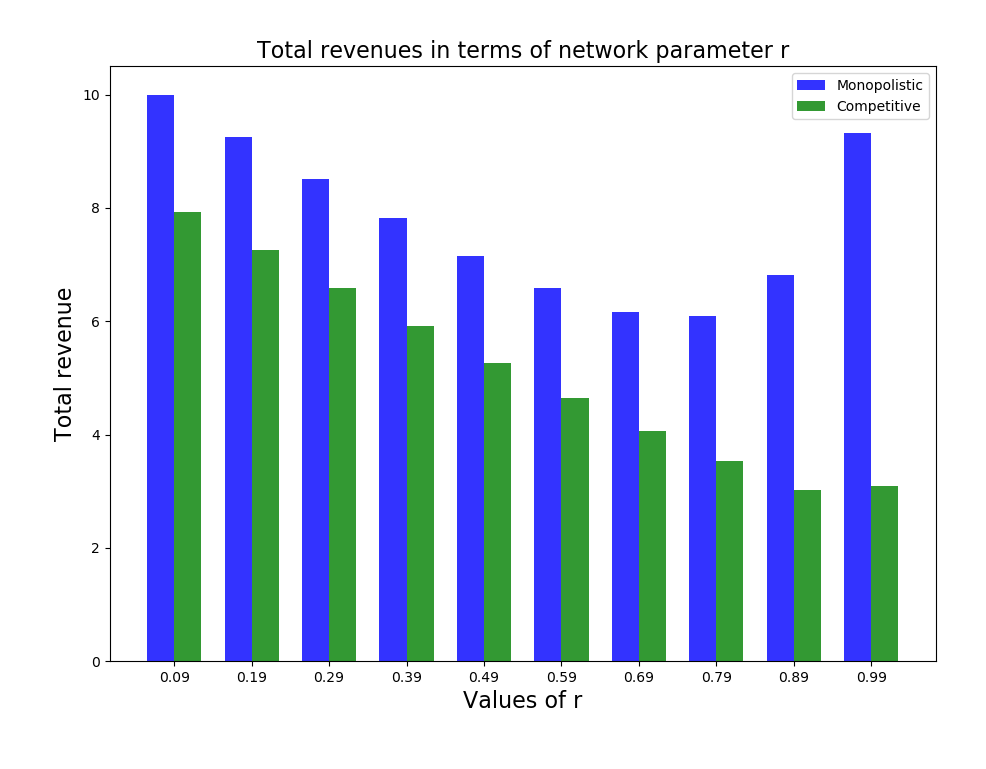}
	\centering{}
	\vspace{-3.5mm}
	\caption{\small Comparison of total revenue perceived by the sellers: competition versus monopoly}
	\label{fig:F3}
	\vspace{-2.5mm}
\end{figure}

Fig. \ref{fig:F3} shows how the total revenue $R(p)$ varies for different values of $r$ (in both, competitive and monopolistic cases). Clearly $R(p^M)>R(p^C)$ for each value of $r$. And $R(p^M)$ has a change of its monotony after some value $r^*>0.7$.

\begin{figure}[ht]
	\vspace{-3.5mm}
	\begin{centering}
		\includegraphics[scale=0.5]{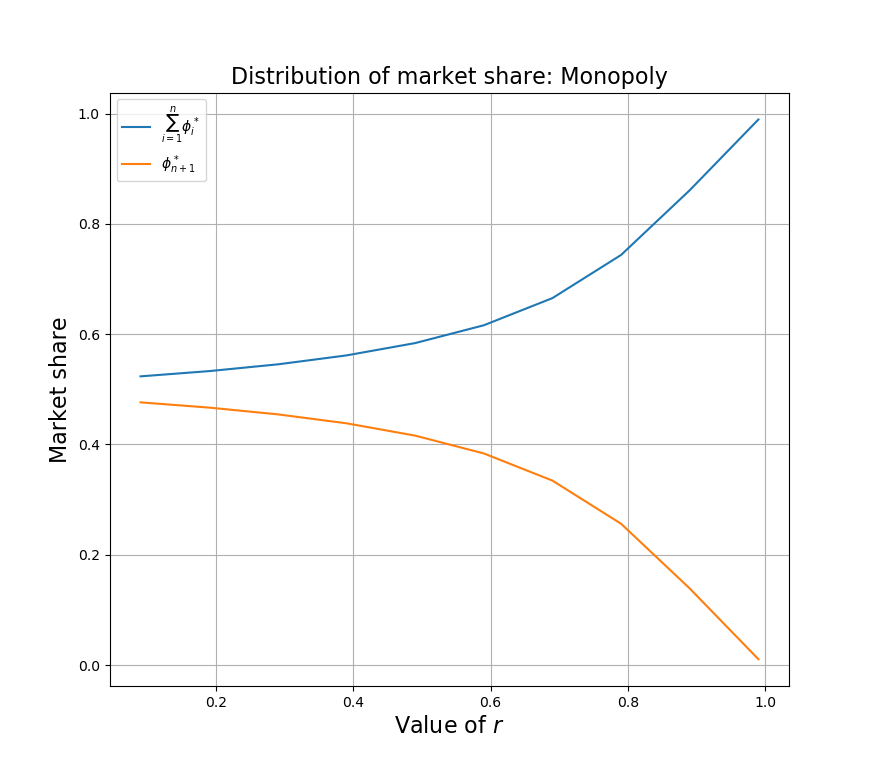}
	\end{centering}
	\centering{}
	\vspace{-3.5mm}
	\caption{\small Comparison of total market shares assigned in the equilibrium for different values of $r$, and the respective market share for the no purchase option, when the monopolistic price is used.}
	\label{fig:F4}
	\vspace{-2.5mm}
\end{figure}

Fig. \ref{fig:F4} depicts the behaviour, for  different values of $r$, of the sum of the market shares for the available products, in contrast to the behaviour of the no purchase option, when the monopolistic price, $p^M$, is used. We can clearly see that $\phi_{n+1}^*$ decreases to zero in terms of $r$, while $\sum_{i=1}^n\phi_i^*$ increases.

\begin{figure}[ht]
	\vspace{-3.5mm}
	\begin{centering}
		\includegraphics[scale=0.5]{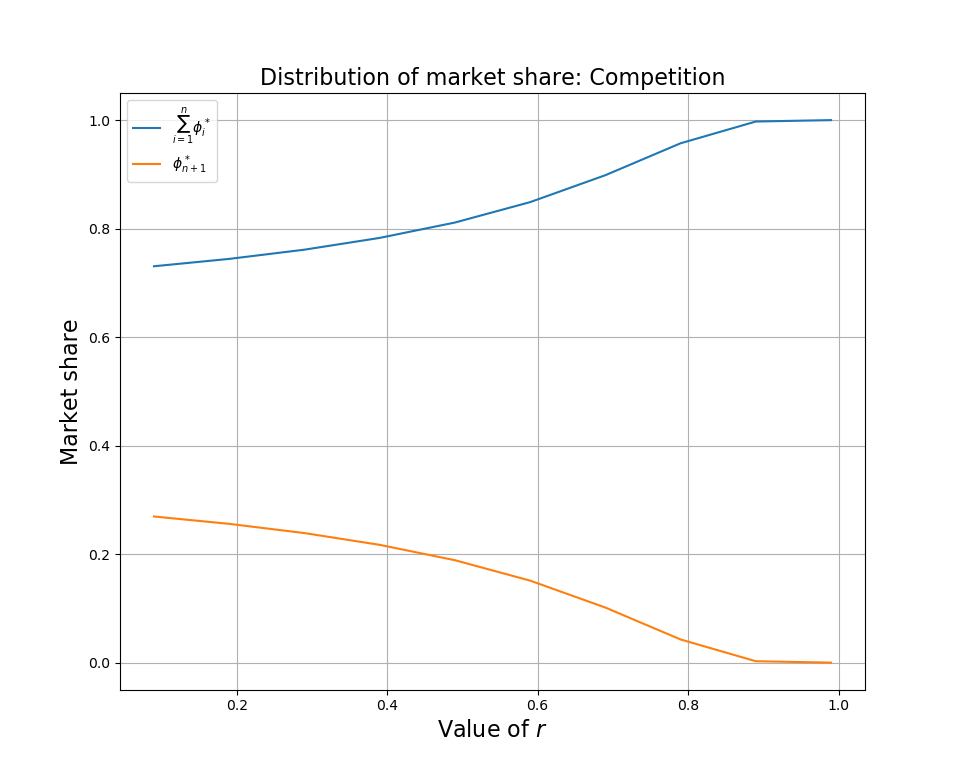}
	\end{centering}
	\centering{}
	\vspace{-3.5mm}
	\caption{\small Comparison of total market shares assigned in the equilibrium for different values of $r$, and the respective market share for the no purchase option, when the competitive price is used.}
	\label{fig:F5}
	\vspace{-2.5mm}
\end{figure}

\end{document}